\documentclass{article}
\usepackage{amsmath} 
\usepackage{amsfonts} 
\usepackage{graphicx}
\usepackage{bm}
\usepackage{tikz}
\usepackage{amsthm}
\graphicspath{ {./Math/} }
\usepackage{tikz-qtree}
\usepackage[sorting = ynt, backend=bibtex, maxnames = 5]{biblatex}
\newtheorem{theorem}{Theorem}[section]
\newtheorem{corollary}[theorem]{Corollary}
\newtheorem{lemma}[theorem]{Lemma}
\theoremstyle{definition}
\newtheorem{definition}[theorem]{Definition}
\begin{document}

\title{Classifying All Transducer Degrees Below $N^3$}
\author{Noah Kaufmann}	

\maketitle

\begin{abstract}
We answer an open question in the theory of transducer degrees initially posed in \cite{3}, on the structure of polynomial transducer degrees, in particular the question of what degrees, if any, lie below the degree of $n^3$. Transducer degrees are the equivalence classes formed by word transformations which can be realized by a finite-state transducer. While there are no general techniques to tell if a word $w_1$ can be transformed into $w_2$ via an FST, the work of Endrullis et al. in \cite{2} provides a test for the class of spiralling functions, which includes all polynomials. We classify fully the degrees of all cubic polynomials which are below $n^3$, and many of the methods can also be used to classify the degrees of polynomials of higher orders.
\end{abstract}

\section{Introduction}

Finite-state transducers (FSTs) and infinite sequences are both fairly natural objects in mathematics, but not much is known about the structure of the order induced by the transducibility relation on infinite sequences ($\sigma \geq \tau$ if $\sigma$ can be transduced into $\tau$ via an FST). However, there are many parallels between the degrees induced by transducibility and Turing degrees, so it's natural to begin by asking the same questions that have been asked (and answered) for Turing degrees. Many initial results have already been obtained (mostly in \cite{1}), for instance: 

\begin{itemize}

\item The bottom degree \textbf{0} is the set of all ultimately periodic sequences (sequences of the form $uvvvv...$ for some words $u$,$v$) 

\item The partial order of degrees is not dense or well-founded

\item There are no maximal degrees 

\item A set of degrees has an upper bound if and only if it is countable 

\item There are at least $\omega$ atom degrees (degrees which are exactly one step above \textbf{0})

\end{itemize}

The last result in particular was proven in (Theorem 32, \cite{4}) when it was shown that there is a way to construct an atomic degree $\langle p_k(n) \rangle$ for every $k$, where $p_k(n)$ is a polynomial of order $k$. (To avoid confusion between degrees of transducibility and degrees of polynomials, we use the term "order" for the degree of a polynomial.) However, until now, all that was known about the structure of degrees of polynomials is that polynomials of different orders are always incomparable, and that the degree of $n^k$ is strictly greater than the degree of $p_k(n)$. We will provide a full classification of all degrees below the degree of $n^3$, with methods that will be useful for classifying degrees of higher order polynomials.

\section{Preliminaries}

Let $\Sigma$ be an alphabet. Recall that the empty word is denoted by $\epsilon$. $\Sigma^*$ is the set of all finite words over $\Sigma$, and $\Sigma^+$ is $\Sigma^*$ \textbackslash $\epsilon$. Denote by $\Sigma^{\mathbb{N}}$ the set of infinite sequences over $\Sigma$. $\Sigma^{\infty}$ is the set of all finite and infinite words over $\Sigma$, i.e. $\Sigma^* \cup \Sigma^{\mathbb{N}}$. For $f: \mathbb{N} \rightarrow \mathbb{N}$ and k in $\mathbb{N}$, the $k$-th shift of f is denoted by $S^k(f)(n) = f(n+k)$ for all $n$ in $\mathbb{N}$.

We use a boldface letter like $\bm{\alpha}$ to denote an ordered tuple $\langle a_0, a_1, \ldots , a_{k-1} \rangle$, and we use $\bm{\alpha^\prime}$ to denote the cyclic shift of $\bm{\alpha}$: $\langle a_1, a_2, \ldots , a_{k-1}, a_0 \rangle$.

\section{Finite-State Transducers and Transducer Degrees}

We will mostly focus on the definitions and results which are useful for the original material in this paper. For a more complete background of the known results for transducer degrees and their proofs, see \cite{1,2,3,4}.

For the following definitions for FSTs, we restrict ourselves to only considering complete pure sequential FSTs, where for every state $q$ and every input letter $a$, there is exactly one successor state $\delta(q,a)$, and the function realized by such a transducer preserves prefixes (if $u$ is a prefix of $v$, then $f(u)$ is a prefix of $f(v)$). We also restrict ourselves to using only \textbf{2} = \{0,1\} as both the input and output alphabet for convenience.

\begin{definition} A \textit{finite-state transducer} is a tuple $T = \langle Q, q_0,\delta,\lambda \rangle$ where $Q$ is a finite set of states, $q_0 \in Q$ is the initial state, $\delta : Q \times \textbf{2} \rightarrow Q$ is the transition function, and $\lambda : Q \times \textbf{2} \rightarrow \textbf{2}^*$ is the output function.  
\end{definition}

Note that $\delta$ and $\lambda$ can be extended ($\delta : Q \times \textbf{2}^* \rightarrow Q,\lambda : Q \times \textbf{2}^\infty \rightarrow \textbf{2}^\infty$) as follows: \\

$$\delta(q,\epsilon) = q, \delta(q,au) = \delta(\delta(q,a),u), \textrm{ where } q \in Q, a \in \textbf{2}, u \in \textbf{2}^*$$

$$\lambda(q, \epsilon) = \epsilon, \lambda(q, au) = \lambda(q,a) \cdot \lambda(\delta(q,a),u), \textrm{ where } q \in Q, a \in \textbf{2}, u \in \textbf{2}^\infty$$

Intuitively, the above equations simply correspond to inputting a string of letters as opposed to just one letter, and reading the string one letter at a time. So then the function $T: \textbf{2}^\infty \rightarrow \textbf{2}^\infty$ induced by the FST $T$ on (finite or infinite) words is defined by $T(u) = \lambda(q_0,u)$, the value of the output $\lambda$ after inputting $u$, starting at the initial state $q_0$. This allows us to give the following definition:

\begin{definition}  Let $T$ be an FST, and let $\sigma, \tau \in \textbf{2}^{\mathbb{N}}$ be infinite sequences. We say that $T$ $transduces$ $\sigma$ to $\tau$, or that $\tau$ is the $T$-$transduct$ of $\sigma$, if $T(\sigma) = \tau$. In general, for any two infinite sequences $\sigma, \tau$ we say that $\sigma \geq \tau$ if there exists some $T$ so that $T(\sigma) = \tau$.
\end{definition}

This relation $\geq$ is reflexive, and can be shown to be transitive by composition of FSTs (See Lemma 8, \cite{1}). If for some $\sigma, \tau$ we have $\sigma \geq \tau$ but not vice versa, we say $\sigma > \tau$. If we do have $\sigma \geq \tau$ and $\tau \geq \sigma$ then we say that $\sigma \equiv \tau$, and we use $[\sigma]$ to denote the equivalence class of $\sigma$. We call $[\sigma]$ the \textit{degree} of $\sigma$. 

\section{Weight Products and Useful Results}

\begin{definition} For a function $f: \mathbb{N} \rightarrow \mathbb{N}$ we define the sequence $\langle f \rangle \in \textbf{2}^{\mathbb{N}}$ by

\begin{center}$\langle f \rangle = \prod_{i=0}^{\infty} 10^{f(i)} = 10^{f(0)}10^{f(1)}10^{f(2)} \ldots$\end{center} 
\end{definition}

Having defined $\langle f \rangle$ in this way, some relatively simple initial results have been obtained in \cite{2}, which we state here.

\begin{lemma} Let $f: \mathbb{N} \rightarrow \mathbb{N}, a,b \in \mathbb{N}.$ We have the following equivalences and inequalities:\\

\begin{enumerate}
\item $\langle af(n) \rangle \equiv \langle f(n) \rangle, $ for $ a > 0$\\
\item $\langle f(n + a) \rangle \equiv \langle f(n) \rangle$\\
\item $\langle f(n) + a \rangle \equiv \langle f(n) \rangle$\\
\item $\langle f(n) \rangle \geq \langle f(an) \rangle, $ for $ a > 0$\\
\item $\langle f(n) \rangle \geq \langle af(2n) + bf(2n + 1) \rangle$\\
\end{enumerate}

\end{lemma}

One interesting consequence of the third equality is that any polynomial with a positive leading coefficient can be thought of as a stream and thus associated with a transducer degree, even if some of its values happen to be negative. For instance, the polynomial $(n-2)^3$ can't directly be interpreted as a stream, since it is negative for $n = 0, 1$. However, if we take $(n-2)^3 + 8$, then this polynomial is nonnegative, and therefore corresponds to a stream (and thus a degree). So even though it's technically incorrect, it will be convenient sometimes to refer to a degree such as $\langle (n-2)^3 \rangle$, when we mean more precisely the degree $\langle (n-2)^3 + k \rangle$ for any $k \geq 8$. Similarly, the first equality allows us to refer to the degree of a function with rational coefficients, where we really mean the degree of the corresponding function multiplied by the appropriate scalar to eliminate any fractional coefficients.

Now we move on to weight products. For this paper we don't provide the full context or proof for the main result we need (Theorem 4.4, which provides us with an incredibly useful result for comparing polynomial degrees), but a more detailed explanation can be found in (Theorem 21, \cite{3}). We start by defining a weight.

\begin{definition} A \textit{weight} is a tuple $\alpha = \langle a_0, a_1, \ldots , a_{k-1},b \rangle \in \mathbb{Q}$ with each $a_i \geq 0$. If $a_i = 0$ for all $i$ then we say the weight is constant. To distinguish between weights and tuples of weights, weights will not be bolded but tuples of weights will.
\end{definition}

Given a weight $\alpha$ as above and a function $f: \mathbb{N} \rightarrow \mathbb{N}$ we can define $\alpha \cdot f$ as:

$$\alpha \cdot f = a_0f(0) + a_1f(1)+ \ldots + a_{k-1}f(k-1)+b$$

We are ready to define the weighted product. Let $\bm{\alpha} = \langle \alpha_0, \alpha_1, \ldots , \alpha_{m-1} \rangle$ be a tuple of weights, with $\bm{\alpha^\prime}$ being the cyclic shift $\langle \alpha_1, \alpha_2, \ldots , \alpha_{m-1}, \alpha_0 \rangle$.

Then the weighted product of $\bm{\alpha}$ with $f$, written as $\bm{\alpha} \otimes f$, is defined in the following way:

$$(\bm{\alpha} \otimes f)(0) = \alpha_0 \cdot f$$
$$(\bm{\alpha} \otimes f)(n+1) = (\bm{\alpha^\prime} \otimes S^{|\alpha_0|-1}(f))(n)$$ 

Here $S^k(f)(n) = f(n+k)$ and $|\alpha_0|$ indicates the length of the tuple $\alpha_0$. We call a weighted product $natural$ if $\bm{\alpha} \otimes f(n) \in \mathbb{N}$ for all $n$. Note that since $\bm{\alpha}$ is a finite tuple of finite tuples in $\mathbb{Q}$, we can take the LCM of all of the denominators and multiply through to make the product natural. Since this does not change the degree of the resulting function (by Lemma 4.2.1), from now on we will assume that all weight products are natural.

The following image provides a more intuitive picture of how the weight product works, by showing pictorially how to compute the weight product of the tuple of weights $\bm{\alpha} = \langle \alpha_0, \alpha_1 \rangle$ with an arbitrary function $f(n)$, where $\alpha_0 = \langle 1,3,3,7 \rangle, \alpha_1 = \langle 4,2,3 \rangle$:

\begin{enumerate}
\item[] 
\begin{tikzpicture}[grow=up]
\Tree [.{$(\bm{\alpha} \otimes f)(0) = f(0)+3f(1)+3f(2)+7$} [.{$\times 3$} $f(2)$ ] [.{$\times 3$} $f(1)$ ] [.{$\times 1$} $f(0)$ ] ]
\end{tikzpicture}
\hskip 0.1in
\begin{tikzpicture}[grow=up]
\Tree [.{$(\bm{\alpha} \otimes f)(1) = 4f(3)+2f(4)+3$} [.{$\times 2$} $f(4)$ ] [.{$\times 4$} $f(3)$ ] ]
\end{tikzpicture}
\end{enumerate}

\begin{enumerate}
\item[] 
\begin{tikzpicture}[grow=up]
\Tree [.{$(\bm{\alpha} \otimes f)(2) = f(5)+3f(6)+3f(7)+7$} [.{$\times 3$} $f(7)$ ] [.{$\times 3$} $f(6)$ ] [.{$\times 1$} $f(5)$ ] ]
\end{tikzpicture}
\hskip 0.1in
\begin{tikzpicture}[grow=up]
\Tree [.{$(\bm{\alpha} \otimes f)(3) = 4f(8)+2f(9)+3$} [.{$\times 2$} $f(9)$ ] [.{$\times 4$} $f(8)$ ] ]
\end{tikzpicture}
\end{enumerate}

Now that we have defined weight products we can proceed to state the main result that we need for the rest of the paper. The full proof of this result can be found in (Theorem 21, \cite{3}).

\begin{theorem} Let $f,g: \mathbb{N} \rightarrow \mathbb{N}$ be polynomials. Then $\langle g \rangle \geq \langle f \rangle$ if and only if there exists a tuple of weights $\bm{\alpha}$ and integers $n_0,m_0$ such that $S^{n_0}(f) = \bm{\alpha} \otimes S^{m_0}(g)$.
\end{theorem}

The usefulness of this theorem for comparing polynomial degrees cannot be overstated. Instead of being forced to consider all possible transductions of polynomial streams by any possible transducer, we can instead narrow our focus strictly to weight products. We also have a similar theorem from \cite{2} which allows us to narrow our focus to only polynomial degrees. The following theorem guarantees that all degrees (not just polynomial degrees) below $n^3$ are equivalent to the degree of a weight product of $n^3$. Therefore we do not miss anything by only considering weight products of $n^3$.

\begin{theorem}
Let $f : \mathbb{N} \rightarrow \mathbb{N}$ be a polynomial, and $\sigma \in 2^\mathbb{N}$. Then $\langle f \rangle \geq \sigma$ if and
only if $\sigma \equiv \langle \alpha \otimes S^{n_0}(f) \rangle$ for some integer $n_0 \geq 0$, and a tuple of weights $\alpha$.
\end{theorem}

We can narrow our focus even further with a new result that lets us focus on a particular subset of weight tuples, the tuples with only one weight. Note that by the definition of weight product, if we take the product of a function $f$ with a weight tuple with a single weight $\bm{\alpha} = \langle a_0, a_1, \ldots , a_{k-1},b \rangle$, it is equal to the sum $b+\sum_{i = 0}^{k-1} a_if(kn+i)$.

\begin{lemma} Let $f(n) = n^k$ and suppose $g: \mathbb{N} \rightarrow \mathbb{N}$ is a polynomial such that there exists a tuple of weights $\bm{\alpha}$ with $g = \bm{\alpha} \otimes f$. Then there exists a single weight $\beta$ such that $g = \langle \beta \rangle \otimes f$.
\end{lemma}

\begin{proof} Let $f$, $g$, $\bm{\alpha}$ be as above, and suppose $\bm{\alpha} = \langle \alpha_0, \alpha_1, \ldots , \alpha_{m-1} \rangle $ has more than one weight. By shifting $g$ as necessary, we can assume WLOG that $\alpha_0 = \langle a_0, a_1, \ldots , a_p,b \rangle$ has the shortest length of all of the weights in $\bm{\alpha}$. Let $L = \sum_{i = 0}^{m-1} (|\alpha_ i| - 1)$.

We have for all $n \equiv 0$ mod $m$, 
$$(\bm{\alpha} \otimes f)(n) = a_0f(\frac{Ln}{m})+a_1f(\frac{Ln}{m}+1)+ \cdots +a_pf(\frac{Ln}{m}+p)+b$$
$$= a_0(\frac{Ln}{m})^k+a_1(\frac{Ln}{m}+1)^k+ \cdots +a_p(\frac{Ln}{m}+p)^k+b$$
$$= \frac{1}{m^k}(a_0(Ln)^k+a_1(Ln+m)^k+ \cdots +a_p(Ln+pm)^k+bm^k)$$

Now by assumption, $\alpha_0$ had the shortest length of the weights in $\bm{\alpha}$, which means that if we take this length (which is $p+1$) times the number of weights $m$, this must be less than or equal to $L$. Therefore we can construct the following weight $\beta$: \\

$\beta = \frac{1}{m^k}(a_0, 0, 0, \cdots , a_1, 0, 0, \cdots, 0, 0, a_p, 0, 0, \cdots, 0, bm^k)$ \\

where each $a_i$ is in position $im+1$, and after $a_p$ we have $L - (pm+1)$ zeroes, so that this weight has length L. Then for all $n \equiv 0$ mod $m$, $(\langle \beta \rangle \otimes f)(n) = (\langle \alpha \rangle \otimes f)(n) = g(n)$. However, all three expressions in this equation are polynomials. Since they are equal for infinitely many $n$, they must be equal for all $n$, and thus $\beta$ is indeed the desired weight.
\end{proof}

Now we are ready to introduce a key definition that forms the basis for all of the results in this paper. 

\begin{definition} Let $\alpha = \langle a_0, a_1, \ldots , a_{k-1},b \rangle$ be a weight with $a_i > 0$ for exactly $m$ $a_i$'s. Then we call the weight product of $\bm{\alpha} = \langle \alpha \rangle$ with a function $f: \mathbb{N} \rightarrow \mathbb{N}$ an \textit{m-transform} of $f$. We will often write $\alpha \otimes f$ instead of $\bm{\alpha} \otimes f$, since Lemma 4.6 guarantees that for this paper we only need to consider 1-element tuples of weights.
\end{definition}

With this definition, we could restate Lemma 4.6 as "All degrees below $\langle n^k \rangle$ are equivalent to an $m$-transform of a shift of $n^k$ for some $m$ $\in \mathbb{N}$".
For this paper we only need to consider the case $k = 3$, and thus if we are interested in classifying all degrees below $\langle n^3 \rangle$, we need to classify the degrees of all $m$-transforms of shifts of $n^3$.

In fact, we can classify almost all $m$-transforms of $n^3$ in one theorem, a generalization of Theorem 32 in \cite{4}. We first need to introduce a couple of definitions from \cite{3}:

\begin{definition} Let $f(n) = a_kn^k+\cdots+a_1n+a_0$ be a polynomial. We define $V(f(n))$ to be the column vector $\begin{bmatrix} a_1 \hspace{.1cm} a_2 \hspace{.1cm} \cdots \hspace{.1cm} a_k \end{bmatrix}^T$. We also define for a weight $\alpha = \langle \alpha_0, \alpha_1, \cdots, \alpha_{m-1},b \rangle$ another column vector $U(\alpha)$, given by $U(\alpha) = \begin{bmatrix} \alpha_0 \hspace{.1cm} \alpha_1 \hspace{.1cm} \cdots \hspace{.1cm} \alpha_{m-1} \end{bmatrix}^T$. Finally, we define $M_m(f(n))$ to be the $k$ by $m$ matrix whose $i$th column is given by $V(f(mn+i-1))$.
\end{definition}

The key feature of these definitions is that it allows us to rewrite the weight product as a matrix product in the following sense:

$$M_m(f(n))U(\alpha) = V((\alpha \otimes f)(n))$$

The proof can be done by checking the definition of a weight product (see Lemma 30 in \cite{4} for more details). Now we need a few more lemmas to set up our first theorem on $m$-transforms.

\begin{lemma} For all $m,a,b,c \in \mathbb{N}$, with $a,b,c$ pairwise distinct and $m \neq 0$, the vectors $V((mn+a)^3),V((mn+b)^3),V((mn+c)^3)$ form an invertible matrix.
\end{lemma}

\begin{proof} The determinant of this matrix is $9m^6(a-b)(b-c)(a-c)$, which is indeed nonzero if $a,b,c$ are distinct and $m \neq 0$ as assumed.
\end{proof}

\begin{lemma} Let q(n) be a polynomial of order k. Then for every $\epsilon > 0$ we have $\langle q(n) \rangle \geq \langle n^k+b_{k-1}n^{k-1}+ \cdots + b_1n \rangle$ for some rational coefficients $0 \leq b_{k-1}, \cdots, b_1 < \epsilon$.
\end{lemma} 

We refer to the polynomial $n^k+b_{k-1}n^{k-1}+ \cdots + b_1n$ that can be found by this lemma as $q_{\epsilon}(n)$, as in \cite{3}. The following lemma is due to Robert Israel in \cite{5}.

\begin{lemma}
Let M be a 3 by n matrix of rank 3 with $n \geq 3$, let $U$ be an n by 1 vector with the first three entries strictly positive and all other entries nonnegative, and let $MU = P$. Let $M_{\epsilon}$ be a matrix the same size as M which has all entries within $\epsilon$ of M. Then there exists a vector $U_{\epsilon}$ such that $M_{\epsilon}U_{\epsilon} = P$, and $U_{\epsilon}$ has the first three entries strictly positive, and all entries nonnegative.
\end{lemma}

\begin{proof}

Since $M$ has rank 3, we can assume WLOG that the first three columns are linearly independent. Then we can write $M$ in block form as $[B \mid C]$, where $B$ is an invertible 3 by 3 matrix, and $C$ is everything else. We can also split $U$ into blocks, and writing these blocks in terms of $B,C,P$ gives us $U = [B^{-1}(P - CW) \mid W]^T$.

Now if $M_{\epsilon} = [B_{\epsilon} \mid C_{\epsilon}]$ is close enough to $M$, $B_{\epsilon}$ is still invertible, and we can define a vector $U_{\epsilon} = [B_{\epsilon}^{-1}(P - C_{\epsilon}W) \mid W]^T$, and by construction $M_{\epsilon}U_{\epsilon} = P$. Further, since $U_{\epsilon}$ converges to $U$, the first three entries can be made strictly positive with a sufficiently small $\epsilon$, and its other entries (the "W" block) are equal to the entries of $U$ and thus nonnegative.  

\end{proof}

\section{Classifying Cubic Degrees}

Now we are ready to begin classifying all cubic degrees below $n^3$. We begin by classifying almost all $m$-transforms of $n^3$ in one result.

\begin{theorem} Let $m \geq 3, t \in \mathbb{Z}$. Then the degrees of every $m$-transform of $(n+t)^3$ are all equal, and every cubic polynomial can be transduced to this degree. We refer to this degree as the bottom degree for cubics.
\end{theorem}

\begin{proof} 
The basic idea of this proof is that we want to narrow our focus to only three columns of our matrix (and three nonzero rows of our vector) in order to replicate the proof method from Theorem 32 in \cite{4}.

Let $f(n) = \sum_{i=0}^{m-1} a_i(mn+t+i)^3$, where each $a_i$ is nonnegative, and at least 3 are strictly positive. By definition we have $M_m((n+t)^3)U(\alpha) = V(f(n))$, where $\alpha$ is the weight $\langle a_0,a_1,...,a_{m-1},0 \rangle$. We will assume for now that it is the first three $a_i$'s which are nonnegative, and justify this assumption later. 

We want to prove that any cubic polynomial can be transduced into $f(n)$. Let $q(n)$ be an arbitrary cubic polynomial, and let $q_{\epsilon}(n)$ be as defined in Lemma 4.9. Consider the matrix $M_m(q_{\epsilon}(n+t))$. We have the right conditions to apply Lemma 4.11 and obtain a vector $U_{\epsilon}$ as stated in the lemma. In particular, we have that $M_m(q_{\epsilon}(n+t))U_{\epsilon}= V(f(n))$ and $U_{\epsilon}$ has all nonnegative entries. But this means that $U_{\epsilon}$ is not just a vector but in fact a weight vector, since we can form a weight out of the entries of $U_{\epsilon}$. Then the equation $M_m(q_{\epsilon}(n+t))U_{\epsilon}= V(f(n))$ corresponds to a weight product, and in fact what it says is that $q_{\epsilon}(n+t)$ can be transduced into $f(n)$ via the weight corresponding to the vector $U_{\epsilon}$. Combining this with the way $q_{\epsilon}$ was defined, we have for an arbitrary cubic polynomial $q(n)$, $\langle q(n) \rangle \equiv \langle q(n+t) \rangle \geq \langle q_{\epsilon}(n+t) \rangle \geq \langle f(n) \rangle$. Therefore $f(n)$ is in the bottom degree for cubics.

However, we do still have to justify why we can assume WLOG that our weight had the first three entries strictly positive. In general, the nonzero entries of the weight could be in any location. But the fundamental argument behind Lemma 4.11 works the same regardless, as long as an invertible matrix can be formed by the columns which correspond to the nonzero entries of the vector $U$ (i.e. if $U$'s strictly positive entries were in the $ith, jth, kth$ positions, the $ith, jth, kth$ columns of $M$ need to be linearly independent). If the nonzero entries are in different places we simply have a messier block matrix to draw and more details to write out. And because Lemma 4.9 guarantees that any 3 columns of $M_m((n+t)^3)$ can form an invertible 3 by 3 matrix, we are justified in using our assumption to simplify the proof.

\end{proof}

We also note that this proof would generalize fairly naturally to prove the following statement: "All $m$-transforms of $(n+t)^k$ are in the bottom degree for polynomials of order $k$, for $m \geq k$".

So now we have completely eliminated all but two cases: 1-transforms and 2-transforms of $n^3$. Next we prove that the 2-transforms are also in the bottom degree.

\begin{theorem} All 2-transforms of $n^3$ are in the bottom degree for cubics.
\end{theorem}

\begin{proof} We first observe that all 2-transforms of $n^3$ are of the form $a(pn+r)^3+b(pn+s)^3$, with $a,b$ both positive rational numbers, $0 \leq r,s < p, r \neq s$, and $p,r,s \in \mathbb{N}$.

We need the following equality: 
$$n^3+d(n+j)^3 = a_1(kn+jk-1)^3+a_2(kn-k+1)^3+a_3(kn)^3$$
where $a_1 = \frac{dj(-1+k+jk)}{k^2(jk-1)(jk+k-2)}, a_2 = \frac{dj}{(k+1)k^2(jk+k-2)}, a_3 = \frac{1}{k^3}(1-\frac{d(jk+k-1)}{(k-1)(jk-1)})$\\

This can be easily verified, and holds for all $d,j,k$. We also observe that if $d$ and $j$ are fixed positive numbers, for a large enough value of $k$ each $a_i$ is strictly positive.

Now we proceed to the proof, which essentially consists of applying algebraic operations to this equation to end up with a 2-transform on the left hand side, and a 3-transform (or higher) on the right hand side. Then applying Thm 5.1 lets us conclude both sides of the equation are in the bottom degree for cubics. \\

Let $a(pn+r)^3+b(pn+s)^3$ be a 2-transform of $n^3$. Assume WLOG $s > r$.

We have two cases.

Case 1: $r > 0$

If r $>$ 0, in the previous equation let $d = \frac{b}{a}$, let $j = s - r$, and choose an integer k large enough so each $a_i$ is positive.Then on both sides of the equation replace $n$ with $(n+r)$ to obtain:\\

$(n+r)^3+d(n+j+r)^3 = a_1(kn+rk+jk-1)^3+a_2(kn+rk-k+1)^3+a_3(kn+rk)^3$\\

Now we use that $r+j = s$, and also replace $n$ with $pn$:\\

$(pn+r)^3+d(pn+s)^3 = a_1(pkn+sk-1)^3+a_2(pkn+(r-1)k+1)^3+a_3(pkn+rk)^3$\\

Multiplying through by $a$ gives us what we want on the left hand side:\\

$a(pn+r)^3+b(pn+s)^3 = aa_1(pkn+sk-1)^3+aa_2(pkn+(r-1)k+1)^3+aa_3(pkn+rk)^3$\\

Now let's look at the right hand side. Since $p > s$ and $s > 1$, $pk > sk-1 > 0$. Since $r > 0$ and $r$ is an integer, $r \geq 1$, so $pk > (r-1)k+1 > 0$. Finally, $p > r$ so $pk > rk > 0$.

Then we can rewrite the right hand side as the following: $\sum_{i=0}^{pk-1} c_i(pkn+i)^3$, where $c_i = 0$ for all i except for $c_{sk-1} = aa_1, c_{(r-1)k+1} = aa_2, c_{rk} = aa_3$. Since $a$ and the $a_i$'s are all positive, this is 3-transform of $n^3$. Therefore it is in the bottom degree, and so is the left hand side.\\

Case 2: $r = 0$ 

This will proceed similarly to the first case, except at the end we get a transformation of $(n-k)^3$. Again let $d = \frac{b}{a}$, $j = s$, and choose an integer k large enough so each $a_i$ is positive. We have:\\

 $n^3+d(n+s)^3 = a_1(kn+sk-1)^3+a_2(kn-k+1)^3+a_3(kn)^3$\\

Now rewrite the right hand side as follows:\\

$n^3+d(n+s)^3 = a_1(kn+(sk+k-1)-k)^3+a_2(kn+1-k)^3+a_3(kn+k-k)^3$\\

Multiply both sides by $a$ and replace $n$ with $pn$ to obtain the desired left hand side:\\

$a(pn)^3+b(pn+s)^3 = aa_1(pkn+(s+1)k-1-k)^3+aa_2(pkn+1-k)^3+aa_3(pkn+k-k)^3$\\

Again we examine the right hand side. Since $p > s$ and both are integers, $p \geq s+1$ and thus $pk > (s+1)k-1 > 0$, and we also have $pk > 1 > 0$, $pk > k > 0$. So we can rewrite the right hand side into the following sum: $\sum_{i=0}^{pk-1} c_i(pkn+i-k)^3$, where $c_i = 0$ for all i except for $c_{(s+1)k-1} = aa_1, c_{1} = aa_2, c_{k} = aa_3$. Since $a$ and the $a_i$'s are all positive, this is a 3-transform of $(n-k)^3$. Therefore it is in the bottom degree by Theorem 5.2, and so is the left hand side.
\end{proof}

Now we have shown that everything except the 1-transforms is equal to the bottom degree for cubics. Fortunately, the 1-transforms are not in the bottom degree, and we will show that their degrees form an interesting structure. First we prove that they are not in the bottom degree.

\begin{theorem} Every 1-transform of $n^3$ is strictly above the bottom degree for cubics.
\end{theorem}

\begin{proof} We will proceed by contradiction. Assume that there is some 1-transform of $n^3$ which is equal to the bottom degree. From the definition of the weight product, all 1-transforms of $n^3$ are of the form $(an+b)^3$, with $a,b \in \mathbb{N}, b < a, a \neq 0$. To say that such a function is in the bottom degree is to say that all other cubic polynomials $f$ satisfy $\langle f \rangle \geq \langle (an+b)^3 \rangle$. In particular, we can choose f to be a 3-transform of $n^3$, say $f = \alpha \otimes n^3$ for some weight $\alpha$ with three nonzero entries.

This means that we have $S^{n_0}((an+b)^3) = \beta \otimes S^{m_0}(\alpha \otimes n^3)$.

A bit of algebra gives us that $\beta \otimes S^{m_0}(\alpha \otimes n^3) = S^{m}(\beta \otimes \alpha \otimes n^3) = S^{m}(\gamma \otimes n^3)$, where $m = m_0*|\beta|$ and $\gamma = \beta \otimes \alpha$ is the composition of weights (Lemma 12, \cite{3}). The key thing to note here is that because $\alpha$ had three nonzero entries, $\gamma$ also must have at least three nonzero entries (The composition of a $k$-transform with an $l$-transform is a $kl$-transform). Let $|\gamma| = j$.\\

Therefore we have:\\

$S^{n_0}((an+b)^3) = S^{m}(\gamma \otimes n^3)$\\

and thus\\

$(an+b+an_0)^3 = \sum_{i=0}^{j-1} c_i(jn+i+jm)^3$\\

where at least 3 of the $c_i's$ are nonzero. We can simplify this equation with the goal of getting just $n^3$ on the left hand side. First we multiply both sides by $j^3$, as well as multiplying the right hand side by $\frac{a^3}{a^3}$ and distributing the top $a^3$ to the inside of the parentheses. This gives:\\

$(ajn+jb+ajn_0)^3 = \sum_{i=0}^{j-1} \frac{j^3}{a^3}c_i(ajn+ai+ajm)^3$\\

We can replace $n$ with $\frac{n}{aj}$ to obtain:\\

$(n+jb+ajn_0)^3 = \sum_{i=0}^{j-1} \frac{j^3}{a^3}c_i(n+ai+ajm)^3$\\

Now replace $n$ with $n-jb-ajn_0$:\\

$n^3 = \sum_{i=0}^{j-1} \frac{j^3}{a^3}c_i(n+ai+ajm-jb-ajn_0)^3$\\

Replace $n$ with $ajn$, and then divide both sides by $a^3j^3$ (this leaves the left hand side unchanged):\\

$n^3 = \sum_{i=0}^{j-1} \frac{1}{a^6}c_i(ajn+ai+ajm-jb-ajn_0)^3$\\

Finally, define $c'_{ak} = \frac{c_k}{a^6}$ for all $k < j$ and 0 otherwise, and let $ajm-jb-ajn_0 = t$.\\

Then we have:\\

$n^3 = \sum_{i=0}^{aj-1} c'_i(ajn+i+t)^3$\\

Now the right hand side is clearly equal to a weight transformation of $(n+t)^3$, and since at least three of the $c_i$'s were nonzero, at least three of the $c'_i$'s will be nonzero. This would mean that $n^3$ is equal to a function which is in the bottom degree for cubics, by Thm 5.1. However, it has been shown that this is not the case in (Theorem 23, \cite{4}), therefore this contradiction means that $\langle (an+b)^3 \rangle$ is not equal to the bottom degree.
\end{proof}

So now we have shown that all of the 1-transforms of $n^3$ are strictly above the bottom degree. We still need to consider how they relate to each other. The following lemma will give us a step in that direction by narrowing down the amount of 1-transforms we need to consider.

\begin{lemma} Suppose $a$ is relatively prime to $b$. Then $\langle (an+b)^3 \rangle \equiv \langle (an+1)^3 \rangle$.
\end{lemma}

\begin{proof} We prove this via double inequality. One direction is trivial:

$\langle (an+b)^3 \rangle \geq \langle (abn+b)^3 \rangle = \langle b^3(an+1)^3 \rangle \equiv \langle (an+1)^3 \rangle$

The other direction requires a bit more work. First, we prove the following: For all $i > 1 \in {N}$, $\langle (an+b^i)^3 \rangle \geq \langle (an+b)^3 \rangle$

Proof: $\langle (an+b^i)^3 \rangle \geq \langle (ab^{i-1}n+b^i)^3 \rangle = \langle b^{3(i-1)}(an+b)^3 \rangle \equiv \langle (an+b)^3 \rangle$

Now we want to show that for some $i$, $\langle (an+1)^3 \rangle \equiv \langle (an+b^i)^3 \rangle$. Since $b$ is relatively prime to $a$, it is invertible in the multiplicative group $\mathbb{Z}^*_a$, which means for some $i$, $b^i \equiv 1 $ mod $ a$. So there exists some number $m$ with $am+1 = b^i$. Then we have:

$\langle (an+1)^3 \rangle \equiv \langle (a(n+m)+1)^3 \rangle = \langle (an+am+1)^3 \rangle = \langle (an+b^i)^3 \rangle$

Since we proved that $\langle (an+b^i)^3 \rangle \geq \langle (an+b)^3 \rangle$ earlier, this completes the proof.
\end{proof}

Note that in the case that $a,b$ are not relatively prime, we can simply divide out by their common factor and reduce. For example, $\langle (4n+2)^3 \rangle = \langle 8(2n+1)^3 \rangle \equiv \langle (2n+1)^3 \rangle$. 

This means that the degrees of all 1-transforms can be represented by some $(an+1)^3$. The only question now is whether or not some (or all) of these degrees are equal, and how they compare to each other. The following theorem answers this question definitively.

\begin{theorem} $\langle (an+1)^3 \rangle \geq \langle (bn+1)^3 \rangle$ if and only if $a | b$.
\end{theorem}

\begin{proof} The reverse implication is relatively straightforward: if $a | b$, then let $b = ac$ and we have $\langle (an+1)^3 \rangle \geq \langle (acn+1)^3 \rangle \equiv \langle (bn+1)^3 \rangle$.

Now for the forward implication. The key observation is that if $\langle (an+1)^3 \rangle \geq \langle (bn+1)^3 \rangle$, this means that there exists $n_0,m_0$, and a weight $\alpha$ such that $S^{n_0}(bn+1)^3 = S^{m_0}(\alpha \otimes (an+1)^3)$, and this weight $\alpha$ must have only one nonzero entry. Otherwise, the right hand side of the previous equation would be an $m$-transform of $n^3$ for $m \geq 2$, which (by Thms 5.1 and 5.2) would put it in the bottom degree. ($(an+1)^3$ is a 1-transform of $n^3$, and an $m$-transform of a 1-transform is also an $m$-transform of the original function, via weight product composition.) To say that $\alpha$ must have only one nonzero entry means that $\alpha \otimes (an+1)^3$ has the form $(a(cn+d)+1)^3 = (acn+ad+1)^3$. Plugging this into the above equation means that\\

$(bn+1+bn_0)^3 = (acn+ad+1+acm_0)^3$\\

This means that $ac = b$ and therefore $a | b$.  
\end{proof}

This theorem also gives us an interesting corollary that is somewhat counterintuitive.

\begin{corollary} Suppose $\langle f \rangle \equiv \langle g \rangle$, and let $\alpha$ be a weight. Then it is possible that $\langle \alpha \otimes f \rangle \not\equiv \langle \alpha \otimes g \rangle$.
\end{corollary}

\begin{proof} Let $f = n^3$, $g = (n+1)^3$, $\alpha = \langle 1,0,0 \rangle$. Then $\langle \alpha \otimes f \rangle = \langle (2n)^3 \rangle \equiv \langle n^3 \rangle$, but $\alpha \otimes g = (2n+1)^3$, and $\langle n^3 \rangle > \langle (2n+1)^3 \rangle$ by the theorem.\\
\end{proof}

\section{Conclusion}

What Theorem 5.5 tells us is that the structure of the degrees of 1-transforms of $n^3$ is that of a divisibility lattice, with $\langle n^3 \rangle$ at the top, $\langle (pn+1)^3 \rangle$ in the second row for all primes $p$, etc, and below this lattice is the bottom degree. The results of Theorems 4.5, 5.1, and 5.2 tell us that these are the only degrees that are below the degree of $n^3$, and thus the structure of all transducer degrees below $n^3$ is given by the following picture (we remove the notation $\langle p(n) \rangle$ used for degrees to make the image cleaner):\\

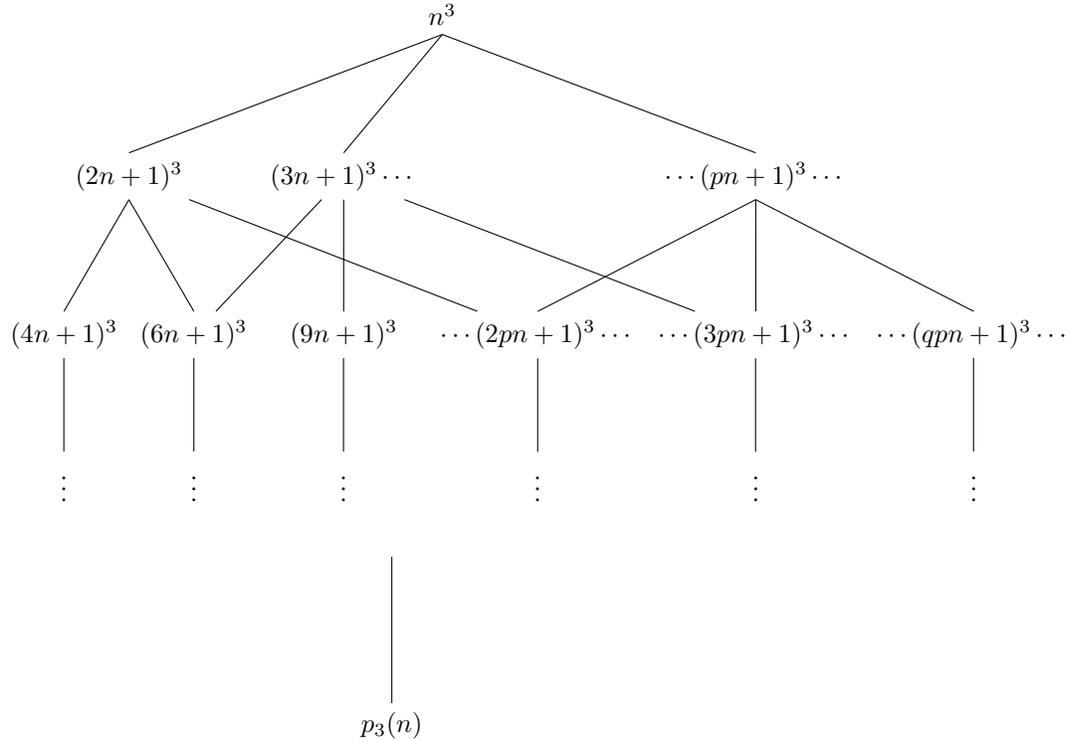
\begin{tikzpicture}
\tikzset{level distance=60pt}
\Tree [.{$n^3$} 
[.\node(2){$(2n+1)^3$}; [.\node(4){$(4n+1)^3$}; {\vdots} ] [.\node(6){$(6n+1)^3$}; {\vdots} ] ] 
[.\node(3){$(3n+1)^3 \cdots$}; [.\node(9){$(9n+1)^3$}; {\vdots} ]]
[.\node(p){$\cdots (pn+1)^3 \cdots$}; [.\node(2p){$\cdots (2pn+1)^3 \cdots$}; {\vdots} ] [.\node(3p){$\cdots (3pn+1)^3 \cdots$}; {\vdots} ] [.\node(qp){$\cdots (qpn+1)^3 \cdots$}; {\vdots} ]]
] 
\draw[-] (3)--(6);
\draw[-] (2)--(2p);
\draw[-] (3)--(3p);
\end{tikzpicture}

\begin{center}
\begin{tikzpicture}[grow'=up]
\hspace*{-.35cm}
\tikzset{level distance=70pt, sibling distance = 50pt}
\Tree [.{$p_3(n)$} [.{} ] ]
\end{tikzpicture}
\end{center}

\begin{center}
$\langle p_3(n) \rangle$ = the bottom degree for cubics (includes all 2-transforms of $n^3$ and all $k$-transforms of $(n+t)^3$ for all $k \geq 3$ and all $t \in \mathbb{N}$)
\end{center}

Therefore all degrees below $n^3$ have been classified. It remains an open question whether or not there are other cubic polynomials that are either above $n^3$ or incomparable with $n^3$. The classification of higher order polynomial degrees is also unknown, but many of the methods used in this paper would be useful to answer such questions.

\printbibliography

\end{document}